\documentclass[proceedings]{stacs} 

\stacsheading{2010}{513-524}{Nancy, France}
\firstpageno{513}

\usepackage{epsf,psfrag,wrapfig,graphics,graphicx,mdwlist,verbatim}

\begin{document}
\sloppy
\title[Approximate shortest paths avoiding a failed vertex]{Approximate shortest paths avoiding a failed vertex : optimal
size data structures for unweighted graphs}

\author[aut1]{N. KHANNA}{NEELESH KHANNA}
\address[aut1]{Oracle India Pvt. Ltd, Bangalore-560029, India.}
\email{neelesh.khanna@gmail.com}

\author[aut2]{S. BASWANA}{SURENDER BASWANA}
\address[aut2]{Indian Institute of Technology Kanpur, India.}
\email{sbaswana@cse.iitk.ac.in}

\thanks{Part of this work was done while the authors were at Max-Planck
Institute for Computer Science, Saarbruecken, Germany during the period 
May-July 2009.}


\begin{abstract}
Let $G=(V,E)$ be any undirected graph on $V$ vertices and
$E$ edges. A path $\textbf{P}$ between any two
vertices $u,v\in V$ is said to be $t$-approximate shortest path if its length
is at most $t$ times the length of the shortest path between $u$ and $v$.
We consider the problem of building a compact data structure for a
given graph $G$ which is capable of answering the following query for
any $u,v,z\in V$ and $t>1$.

\centerline{\em report $t$-approximate shortest path between $u$ and $v$ when
vertex $z$ fails}

We present data structures for the single source as well all-pairs versions
of this problem. Our data structures guarantee optimal query time. Most
impressive feature of our data structures is that their size {\em nearly} match
the size of their best static counterparts.
\end{abstract}

\maketitle

\section{Introduction}
The shortest paths problem is a classical and well studied algorithmic problem
of computer science. This problem requires processing of a given graph
$G=(V,E)$ on $n=|V|$ vertices and $m=|E|$ edges to compute a data structure
using which shortest path or distance between any two vertices can be
efficiently reported. Two famous and thoroughly studied versions of this problem
are single source shortest paths (SSSP) problem and all-pairs shortest
paths (APSP) problem.

Most of the applications of the shortest paths problem involve real
life graphs and networks which are prone to failure of nodes (vertices) and
links (edges).
This has motivated researchers to design {\em dynamic} solution for the shortest
paths problem. For this purpose, one has to first develop a suitable model
for the shortest paths problem in dynamic networks. In fact two such models
exists, and each of them has its own algorithmic objectives.

The shortest paths problem in the first model is described
as follows : There is an initial graph followed by an on-line sequence of
insertion and deletion of edges interspersed with shortest path (or distance) 
queries. Each query has to be answered with respect to the
graph which exists at that moment (incorporating all the updates preceding the
query on the initial graph).
A trivial solution of this problem is to recompute all-pairs shortest
paths from scratch after each update. This is certainly a wasteful approach
since a single update usually does not cause a huge change in the
all-pairs distance information. Therefore, the algorithmic objective here
is to maintain a data structure which can answer distance query efficiently
and can be updated after any edge insertion or deletion in an efficient manner.
In particular, the time required to update the data structure has to  be
substantially less than the running time of the best static algorithm.
Many novel algorithms have been designed in the last ten years for this problem
and its variants (see \cite{DI:4} and the references therein).

On one hand the first model is important since it captures the worst possible
hardness of any dynamic graph problem. On the other hand, it can also be
considered as a pessimistic model for real life networks.
It is true that the networks are never immune to
failures. But in addition to it, it is also rare to have networks which may
have arbitrary number of failures in normal circumstances.
It is  essential for network designers to choose suitable technology to
make sure that the failures are quite infrequent in the network.
In addition, when a vertex or edge fails (goes down), it does not
remain failed/down indefinitely. Instead, it comes
up after some finite time due to simultaneous repair mechanism going on in the
network. These aspects can be captured in the second model which takes as input
a graph and a number $\ell \ll n$. This model assumes that there will be at
most $\ell$ vertices or edges which may be inactive at any time, though
the corresponding set of failed vertices or edges may keep changing as the
time progresses : the old failed vertices become active while some new active
vertices may fail. The algorithmic objective in this model is to preprocess the
given graph to construct a compact data structure which for any subset $S$ of
at most $\ell$ vertices may answer the following query for any $u,v\in V$.

\centerline{\em
Report the shortest-path (or distance) from $u$ to $v$ in $G\backslash S$.}

It is desired that each query gets answered in {\em optimal time} : retrieval of
distance in $O(1)$ time and the shortest path in time which is of the order
of the number of its edges. The ultimate research goal would be to understand
the complexity of the above problem for any given value $\ell$. In this
pursuit, the first natural step would be to efficiently solve and thoroughly
understand the complexity of the problem for the case $\ell=1$, that is, 
the shortest paths problem avoiding any failed vertex. Interestingly, this 
problem appears as a sub problem in many other related problems, namely, 
Vickrey pricing of networks \cite{HS:1}, most vital node of a shortest path 
\cite{NPW:3}, the replacement path problem \cite{R:7}, 
and shortest paths avoiding forbidden subpaths 
\cite{ahmed_et_al:LIPIcs:2009:1831}.

The first nontrivial and quite significant breakthrough on the all-pairs
version of this problem was made by Demetrescu et al. \cite{DT:1}. They
designed an $O(n^2\log n)$ space data structure, namely {\em distance 
sensitivity oracle}, which is capable of
reporting the shortest path between any two vertices avoiding
any single failed vertex. The preprocessing time of this data structure
is $O(mn^2)$. Recently, Bernstein and Karger \cite{BK:9} improved
the preprocessing time to $O(mn\log n)$.
Though ${\Theta}(n^2\log n)$ space bound of this all-pairs
distance sensitivity oracle is optimal up to
logarithmic factors, it is too large for many real life graphs which appear
in various large scale applications \cite{TZ:5}. In most of these graphs
usually $m\ll n^2$, hence a table of $\Theta(n^2)$ size may be too large for
practical purposes. However, it is also known \cite{DT:1} that  even a data structure which 
reports exact distances from a fixed source avoiding a single failed vertex will require $\Omega(n^2)$ space in the worst case.
So approximation seems to be the only way to design a small space compact
data structure for the problem of shortest paths avoiding a failed vertex.
A path between $u,v\in V$ is said to be $t$-approximate shortest path if its
length is at most $t$ times that of the shortest path between the two.
The factor $t$ is usually called the stretch. We would like to state here that
many algorithms and data structures have been designed in the last fifteen years
for the static all-pairs approximate shortest paths
(see \cite{BK:6,TZ:5} and references therein). The prime motivation
underlying these algorithms has been to achieve sub-quadratic space and/or
sub-cubic preprocessing time for the static APSP problem.
However, no data structure was designed in the past for approximate shortest
paths avoiding any failed vertex.

In this paper, we present really compact data structures which are capable
of reporting approximate shortest paths between two vertices avoiding any 
failed vertex in undirected graphs. The
most impressive feature of our data structures is their nearly optimal size.
In fact their size almost matches the size of their best static counterparts.
\subsection{New Results}
\noindent
{\bf Single source approximate shortest paths avoiding any failed vertex.}\\
First we address weighted graphs. For the weighted graphs, we present an
$O(m\log n)$ time constructible data structure of size $O(n \log n)$ which
can report 3-approximate shortest path from the source to any vertex $v\in V$
avoiding any $x\in V$.
We then consider the case of undirected unweighted graphs.
For these graphs, we present an $O(n\frac{\log n}{\epsilon^3})$ space data
structure which can even report $(1+\epsilon)$-approximate shortest path
for any $\epsilon>0$.

\noindent
{\bf All-pairs approximate shortest paths avoiding any failed vertex.}\\
Among the existing data structures for static all-pairs approximate shortest
paths, the {\em approximate distance oracle} of Thorup and Zwick \cite{TZ:5}
stands out due to its amazing features.
Thorup and Zwick \cite{TZ:5} showed that an
undirected graph can be preprocessed in sub-cubic time to build a data
structure of size $O(kn^{1+1/k})$ for any $k>1$. This data structure, despite
of its sub-quadratic size, is capable of reporting $(2k-1)$-approximate
distance between any two vertices in $O(k)$ time
(and the corresponding approximate shortest path in optimal time), and hence
the name {\em oracle}. Moreover, the size-stretch trade off achieved by this 
data structure is essentially optimal.
It is a very natural question to explore whether it is possible to design
all-pairs approximate distance oracle which may handle single vertex failure.
We show that it is indeed possible for unweighted graphs. For this purpose, we
suitably modify the approximate distance oracle of Thorup and Zwick
\cite{TZ:5} using some new insights and our single source
data structure mentioned above. These modifications
make the approximate shortest-paths oracle of Thorup and Zwick handle vertex
failure easily, and (surprisingly) still preserving the old (optimal)
trade-off between the space and the stretch. For precise details, see
Theorem \ref{APASP-oracle-under-vertex-failure}.

For the algorithmic details missing in this extended abstract due to page
limitations, we suggest the reader to refer to the journal version
\cite{KB:09}. Our data structures can be easily adapted for handling edge
failure as well without any increase in space or time complexity.

\section{Preliminaries}
We use the following notations and definitions in the context of a given
undirected graph $G=(V,E)$ with $n=|V|$, $m=|E|$ and a weight
function $\omega : E \rightarrow \textbf{R}^+$.
\begin{itemize}
\item
$T_r$ : single source shortest path tree rooted at $r$.
\item $\textbf{P}(x,y)$ : the shortest path between $x$ and $y$.
\item $\delta(x,y)$ : the length of the shortest path between $x$ and $y$.
\item
$\textbf{P}(x,y,z)$ : the shortest path between $x$ and $y$ avoiding vertex $z$.
\item
$\delta(x,y,z)$ : the length of the shortest path between $x$ and $y$ avoiding
vertex $z$.
\item
$T_r(x)$ : the subtree of $T_r$ rooted at $x$.
\item $G_r(x)$ :
the subgraph induced by the vertices of set $T_r(x)$ and augmented by vertex
$r$ and edges from $r$ as follows. For each $v\in T_r(x)$ with
neighbors outside $T_r(x)$, keep an edge $(r,v)$
of weight = $\min_{(u,v)\in E, u\notin T_r(x)}(\delta(r,u)+\omega(u,v))$.
\item $P::Q$ : a path formed by concatenating path $Q$ at the end of path $P$
with an edge $(u,v)\in E$, where $u$ is the last vertex of $P$ and $v$ is the
first vertex of $Q$.
\item
$E(X)$ : the set of edges from $E$ with at least one endpoint in $X$.
\end{itemize}
Our algorithms will also use a data structure for answering
lowest common ancestor (LCA) queries on $T_r$. There exists an $O(n)$ time 
computable data structure which occupies $O(n)$ space and can answer any LCA 
query in $O(1)$ time (see \cite{BF:00} and references therein).

\section{Single source 3-approximate shortest paths avoiding a failed vertex}
We shall first solve a simpler sub-problem where the vertex which may
fail belong to a given path $P \in T_r$. Then we use divide and conquer
strategy wherein we decompose $T_r$ into a set of disjoint paths and for each
such path, we solve this sub-problem.
%
\subsection{Solving the Sub-Problem : the failures of a vertex from a given path $\textbf{P}(r,t)$}
\noindent
Given the shortest path tree $T_r$, let
$\textbf{P}(r,t)=\langle r(=x_0),x_1,...,x_k(=t)\rangle$ be any shortest path
present in $T_r$. We shall design an $O(n)$ space data structure which
will support retrieval of a 3-approximate shortest path from $r$ to any
$v\in V$ when some vertex from $\textbf{P}(r,t)$ fails.
The preprocessing time of our algorithm will be $O(m+n\log n)$ which matches
that of Dijkstra's algorithm. The algorithm is inspired by the algorithm of
Nardelli et al. \cite{NPW:3} for computing the most vital vertex on
a shortest path.
\begin{figure}[h]
\psfrag{r}{$r$}
\psfrag{xi-1}{$x_{i-1}$}
\psfrag{xi}{$x_i$}
\psfrag{xi+1}{$x_{i+1}$}
\psfrag{Oi}{$O_i$}
\psfrag{Ui}{$U_i$}
\psfrag{Tri}{$D_i$}
\psfrag{P}{$\textbf{P}$}
\centerline{\epsfysize=115pt \epsfbox{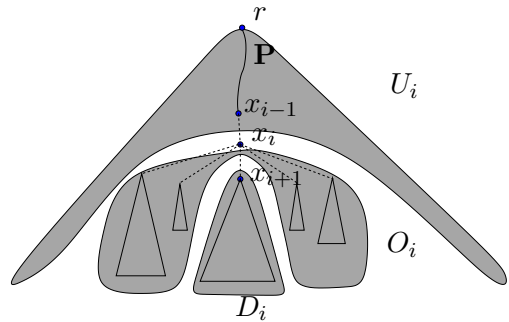}}
\caption{Partitioning of the shortest path tree $T_r$ at $x_i\in \textbf{P}$}
\label{three-parts}
\end{figure}
Consider vertex $x_i$ lying on the path $\textbf{P}(r,t)$. We partition
the tree $T_r\backslash \{x_i\}$ into the following 3 parts (see Figure \ref{three-parts}).
\begin{enumerate*}
\item $U_i$ : the tree $T_r$ after removing the subtree $T_r(x_i)$
\item $D_i$ : the subtree of $T_r$ rooted at $x_{i+1}$
\item $O_i$ : the portion of $T_r$ left after removing $U_i$, $x_i$,
and $D_i$.
\end{enumerate*}
Note that a vertex of the tree $T_r$ is either a vertex of the
path $\textbf{P}(r,t)$ or it belongs to some $O_i$ for some $i$.
We build the following two data-structures of total $O(n)$ size.
\begin{enumerate*}
\item
a data structure to retrieve 3-approximate shortest path from $r$ to any
$v\in D_i$.
\item
a data structure to retrieve 3-approximate shortest path from $r$ to any
$v\in O_i$.
\end{enumerate*}
\subsubsection{Data structure for 3-approximate shortest paths to vertices of
$D_i$ when $x_i$ has failed}
Consider the vertex $x_{i+1}$ and any other vertex $y\in D_i$. Note that the
shortest path $\textbf{P}(x_{i+1},y)$ remains intact even after removal of
$x_i$, and its length is certainly less than $\delta(r,y)$. Based on this simple
observation one can intuitively see that in order to travel from $r$ to $y$
when $x_i$ fails, we may travel along shortest route to $x_{i+1}$
(that is $\textbf{P}(r,x_{i+1},x_i)$) and then along $\textbf{P}(x_{i+1},y)$.
Using triangle inequality and the fact that the graph is undirected, the 
length of this path $\textbf{P}(r,x_{i+1},x_i)::\textbf{P}(x_{i+1},y)$ can
be approximated as follows.
\begin{eqnarray*}
\delta(r,x_{i+1},x_i)+\delta(x_{i+1},y)& \le &
\delta(r,y,x_i)+\delta(y,x_{i+1},x_i)+\delta(x_{i+1},y)\\
          & \le & \delta(r,y,x_i) + 2\delta(x_{i+1},y)\\
          & \le & \delta(r,y,x_i) + 2\delta(r,y) \le  3\delta(r,y,x_i)
\end{eqnarray*}
Therefore, in order to support retrieval of 3-approximate shortest path to any
$v\in D_i$ in optimal time, it suffices to
store the path $\textbf{P}(r,x_{i+1},x_i)$.

In order to devise ways of efficient computation and compact storage
of $\textbf{P}(r,x_{i+1},x_i)$ for a given $i$, we use the following lemma
about the structure of the path $\textbf{P}(r,x_{i+1},x_i)$.
\begin{lemma}
The shortest path $\textbf{P}(r,x_{i+1},x_i)$ is of the form $P_1::P_2$ where
$P_1$ is a shortest path from $r$ in the subgraph induced by $U_i\cup O_i$, and
$P_2$ is a path present in $D_i$.
\label{obs-2}
\end{lemma}
It follows that in order to compute $\textbf{P}(r,x_{i+1},x_i)$, first we need
to compute shortest paths from $r$ in the subgraph induced by $U_i\cup O_i$.
Let $\delta_i(r,v)$ denote the distance from $r$ to
$v\in U_i\cup O_i$ in this subgraph.
Note that $\delta_i(r,v)$ for $v\in U_i$ and the corresponding
shortest path is the same as in the original graph, and is already present in
$T_r$. For computing shortest paths from $r$ to vertices of $O_i$,
we build a shortest path tree (denoted as $T_r(O_i)$) from $r$ in the subgraph
induced by vertices $O_i\cup\{r\}$ and the following additional edges. For each
$z \in O_i$ with at least one neighbor in $U_i$, we add an edge $(r,z)$ with
weight = $min_{(u,z)\in E, u \in U_i}(\delta(r,u) + \omega(u,z))$.
Applying Lemma 3.1, let $(y_i,z_i)$ be the edge of ${\mathbf P}(r,x_{i+1},x_i)$
joining the sub path present in $U_i\cup O_i$ with the sub path present in
$D_i$. This edge can be identified using the fact that this is the edge which
minimizes $\delta_i(r,y)+\omega(y,z)+\delta(x_{i+1},z)$ over all
$z\in D_i,y\in U_i\cup O_i,(y,z)\in E$. The vertex $x_{i+1}$ stores the path
$\textbf{P}(r,x_{i+1},x_i)$ implicitly by keeping the edge $(y_{i},z_{i})$ and
the tree $T_r(O_i)$. The shortest path $\textbf{P}(r,x_{i+1},x_i)$ can
be retrieved in optimal time using the trees $T_r$, $T_r(O_i)$, and the edge
$(y_{i+1},z_{i+1})$. Due to mutual disjointness of $O_i$'s, the overall space
requirement of the data structure for retrieving $\textbf{P}(r,x_{i+1},x_i)$
for all $i\le k$ will be $O(n)$.
%
%
\subsubsection{Data structure for 3-approximate shortest paths to vertices of
$O_i$ when $x_i$ has failed}
In order to compute 3-approximate shortest path to $O_i$ upon failure of $x_i$,
we shall use the approximate shortest paths to $D_i$ as computed above.
Here we use an interesting observation which states that if we have
a data structure to retrieve $\alpha$-approximate shortest paths
from $r$ to vertices of $D_i$ when $x_i$ fails, then we can use it to have a 
data-structure to retrieve $\alpha$-approximate shortest paths to vertices of 
$O_i$ as well. To prove this result,
this is how we proceed. Consider the subgraph induced by $O_i$ and augmented
with vertex $r$ and some extra edges which are defined as follows.
\begin{itemize}
\item
For each $o\in O_i$ having neighbors from $U_i$, keep an edge
$(r,o)$ and assign it weight =
$\min_{(u,o)\in E, u\in U_i}(\delta(r,u) + \omega(u,o))$.
\item
For each $o\in O_i$ having neighbors from $D_i$, keep an
edge $(r,o)$ and assign it weight = $\min_{(u,o)\in E, u\in D_i}(\hat{\delta}(r,u,x_i)+\omega(u,o))$, where $\hat{\delta}(r,u,x_i)$ is the $\alpha$-approximate 
distance to $u$ upon failure of $x_i$. 
(In the present situation we have $\alpha=3$.)
\end{itemize}
Let us denote this graph as $G_r(O_i)$. Observation \ref{generic-t} is based
on the following lemma which is easy to prove.
\begin{lemma}
The Dijkstra's algorithm from $r$ in the graph $G_r(O_i)$ computes
$\alpha$-approximate shortest paths from $r$ to all $v\in O_i$ avoiding $x_i$.
\label{O_i-lemma}
\end{lemma}
\begin{obs}
If we can design a data structure for retrieving $(1+\epsilon)$-approximate
shortest paths from $r$ to vertices of $D_i$ upon failure of $x_i$, then it
can also be used to design a data structure which can support retrieval of
$(1+\epsilon)$-approximate shortest paths to all vertices of the graph
upon failure of $x_i$.
\label{generic-t}
\end{obs}
\noindent
We compute and store the shortest path tree rooted at $r$ in the graph
$G_r(O_i)$. This tree along with the tree $T_r$ and the data structure
described in the previous sub-section suffice for retrieval of 3-approximate
shortest paths to $o\in O_i$ upon failure of $x_i$.

\noindent
{\bf Query answering:}~
Suppose the oracle receives a query asking
for approximate shortest path from $r$ to $v$ avoiding
$x_i\in {\textbf P}(r,t)$. It first invokes lowest common ancestor (LCA) query
between $v$ and $x_i$ on $T_r$. If $LCA(v,x_i)\not=x_i$, the shortest path
from $r$ to $v$ remains unaffected and so it reports the path $\textbf{P}(r,t)$.
Otherwise, it determines if $v\in D_i$ or $v\in O_i$.
Depending upon the two cases, it reports the approximate shortest path between
$r$ and $v_i$ using one of the two data structures described above.
\begin{theorem}
An undirected weighted graph $G=(V,E)$, a source $r\in V$, and a shortest path
$P\in T_r$ can be processed in $O(m+n\log n)$ time to build a data
structure of $O(n)$ space which can report 3-approximate shortest
path from $r$ to any $v\in V$ avoiding any single failed vertex from $P$.
\label{path-theorem}
\end{theorem}

\subsection{Handling the failure of any vertex in $T_r$}
We follow divide and conquer strategy based on the following simple lemma.
\begin{lemma}
There exists an $O(n)$ time algorithm to compute a path $P$ in $T_r$
whose removal splits $T_r$ into a collection of disjoint subtrees
$T_r(v_1),...T_r(v_j)$ such that
\begin{itemize}
\item $|T_r(v_i)| < n/2$ for each $i\le j$.
\item $P \cup_i T_r(v_i) = T$ and $P\cap T_r(v_i)=\emptyset ~~~\forall i$.
\end{itemize}
\label{tree-partition}
\end{lemma}
First we compute the path $P\in T_r$ as mentioned in Lemma
\ref{tree-partition}. We build the data structure for handling failure of any
vertex from $P$ by executing the algorithm of Theorem \ref{path-theorem}.
Let $v_1,...,v_j$ be the roots of the sub trees of $T_r$ connected to the path
$P$ with an  edge. For each $1\le i\le j$, we solve the problem recursively on the subgraph
$G_r(v_i)$, and build the corresponding data structures.
Lemma \ref{tree-partition} and Theorem \ref{path-theorem} can be used in
straight forward manner to prove the following theorem.
\begin{theorem}
An undirected weighted graph $G=(V,E)$ can be processed in
$O(m\log n +n \log^2 n)$ time to build a data structure of size
$O(n\log n)$ which can answer, in optimal time,
any 3-approximate shortest path query from a given source $r$ to any vertex
$v\in V$ avoiding any single failed vertex.
\label{tree-theorem}
\end{theorem}

\section{Single source (1+$\epsilon$)-approximate shortest paths avoiding a
failed vertex}

In this section, we shall present a compact data structure for single source
$(1+\epsilon)$-approximate shortest paths avoiding a failed vertex in
an unweighted graph.
Let $level(v)$ denote the level (or distance from $r$) of vertex $v$ in the tree $T_r$.  Let $U_x,D_x,O_x$ denote the partitions of the
tree $T_r$ formed by deletion of vertex $x$, with the same meaning as that of $U_i,D_i,O_i$
defined for $x_i$ in the previous section.
On the basis of Observation \ref{generic-t},
our objective is to build a compact data structure
which will support retrieval of $(1+\epsilon)$-approximate
shortest-paths to vertices of $D_x$ upon failure of $x$ for any
$x\in V$.
Let ${\tt uchild}(x)$ denote the root of the subtree corresponding to
$D_x$ (it is similar to $x_{i+1}$ in case of $D_i$). For reporting approximate
distance between $r$ and $v\in D_x$ when $x$ fails, the data structure of
previous section reports path of length
$\delta(r,{\tt uchild}(x),x) + \delta({\tt uchild}(x),v)$ which is bounded by
$\delta(r,v,x) + 2\delta({\tt uchild}(x),v)$. It should be noted that the
approximation factor associated with it is already bounded by $(1+\epsilon)$
for any $\epsilon>0$ if the following condition holds.

\centerline{ $\textbf{C}$ : ${\tt uchild}(x)$ is {\em close} to $v$, that is,
$\delta({\tt uchild}(x),v) \le \frac{\epsilon}{2} \delta(r,v)$.}

We shall build a supplementary data structure which will ensure that whenever
the condition $\textbf{C}$ does not hold, there will be some ancestor $w$
of $v$ lying on $\textbf{P}(x,v)$, called a {\em special} vertex, satisfying the
following two properties.
\begin{enumerate}
\item
$\delta(w,v) \ll \delta(r,v)$, that is $w$ is much closer to $v$ than $r$.
\item
vertex $w$ stores approximate shortest path to $r$ avoiding $x$
(with the approximation factor arbitrarily close to 1).
\end{enumerate}
We shall refer to such vertices $w$ as special-vertices.
%
\subsection{Constructing the set of special vertices}

Let $h$ be the height of BFS tree rooted at $r$.
Let  $L$ be a set of integers such that
$L=\{i|\lfloor{(1 + {\epsilon})}^i \rfloor < h \}$.
For a given $i\in L$, we define a subset $S_i$ of special vertices as $S_i = \{
u \in V|
level(u)=\lfloor{(1+{\epsilon})}^i\rfloor \wedge |T_r(u)|\geq{\epsilon} level(u)
\}$. We define the set of special vertices as $S = \cup_{\forall i \in L}S_i$. In
addition, we also introduce the following terminologies.
\begin{itemize}
\item
$S(v)$:~ the nearest ancestor of $v$ which belongs to set $S$.
\item
$V(u)$:~ For a vertex $u\in S$, $V(u)$ denotes the set of vertices $v\in V$ with
$S(v)=u$. In essence, the vertex $u$ will serve as the special vertex for each
vertex from $V(u)$. For failure of any vertex $x\in \textbf{P}(r,u)$, each
vertex of set $V(u)$ will query the data structure stored at $u$ for retrieval
of approximate shortest path/distance from the source.
\end{itemize}

We now state two simple lemmas based on the above construction.
\begin{lemma}
Let $v\in V \backslash S$, then
$\delta(v,S(v)) \leq \big(\frac{2{\epsilon}}{1+{\epsilon}}\big)level(v)$
if ${\epsilon} < 1$
\label{prop1}
\end{lemma}
\begin{lemma}
\label{prop2}
Let $u$ be a vertex at level $\ell$ and $u\in S$. Then
$V(u) \ge {\epsilon} \ell$.
\end{lemma}
If we can ensure that the data structure for a special vertex $u$ (for 
retrieving approximate shortest paths from $r$ upon failure of any 
$x \in \textbf{P}(r,u)$) is of size $O(level(u))$, then it would follow from 
Lemma \ref{prop2} that the space required by our supplementary data structure 
will be linear in $n$.
\subsection{The data structure for a special vertex}
Consider a special vertex $v$ with $level(v) = \lfloor (1+\epsilon)^i \rfloor$
We shall now describe a compact data structure stored at $v$ which will
facilitate retrieval of approximate shortest path from $r$ to $v$ upon failure
of any vertex $x\in \textbf{P}(r,v)$.

Let $v'$ be the special vertex which is present at level
$\lfloor (1+\epsilon)^{i-1}\rfloor$ and is ancestor of $v$.
The data structure stored at $v$ will be defined in a way that will prevent it 
from storing information that is already present in the data structure of some 
special vertex lying on $\textbf{P}(r,v')$.

If $x \in \textbf{P}(v',v)$, then the data structure described in the
previous section itself stores a path which is $(1+2\epsilon)$-approximation
of $\textbf{P}(r,v,x)$.

Let us now consider the nontrivial case when $x\in \textbf{P}(r,v'), x\not=v'$.
In order to discuss this case, we would like to introduce the notion of
{\em detour}. To understand it, let us visualize the paths
$\textbf{P}(r,v,x)$ and $\textbf{P}(r,v)$ simultaneously. Since $\textbf{P}(r,v,x)$ and $\textbf{P}(r,v)$ have the same end-points and $x$ doesn't lie on $\textbf{P}(r,v,x)$, there must be a  {\em middle} portion
of $\textbf{P}(r,v,x)$ which intersects $\textbf{P}(r,v)$ at exactly two vertices, and the remaining portion of $\textbf{P}(r,v,x)$ overlaps with $\textbf{P}(r,v)$. This {\em middle} portion
is called a detour. We now define it more formally. Let $a$ and $b$ be two vertices on the shortest path $\textbf{P}(r,v)$.  We use $a\prec b$ to denote that vertex $a$ is closer to $r$ than
vertex $b$. The notation $a\preceq b$
would mean that either $a\prec b$ or $a=b$. .
So here is the definition of detour (and the underlying observation).
\begin{definition}
Let $x\in \textbf{P}(r,y)$. When $x$ fails, the path $\textbf{P}(r,y,x)$ will be
of the form of $\textbf{P}(r,a)::p_{a,b}::\textbf{P}(b,y)$, where
$r\preceq a \prec x\prec b\preceq y$ and the path $p_{a,b}$ is such that
$p_{a,b}\cap \textbf{P}(a,b)=\{a,b\}$. In other words, $p_{a,b}$ meets
$\textbf{P}(a,b)$ only at the end points. We shall call $p_{a,b}$ as the detour
associated with the shortest path
$\textbf{P}(r,y,x)$.
\end{definition}

Let $p_{a,b}$ represent the detour w.r.t. to $\textbf{P}(r,v,x)$.
The handling of failure of vertices $x\in \textbf{P}(r,v)$ which lie above $v'$
would depend upon the detour $p_{a,b}$. This detour can be of any of the
following types (see Figure  \ref{fig:detour-type} for illustration).
\begin{itemize}
\item I : $b \preceq v'$.
\item II : $v' \prec b$.
\end{itemize}
\begin{figure}[h]
\psfrag{l}{$\ell_0$}
\psfrag{x}{$x$}
\psfrag{x1}{$a$}
\psfrag{x2}{$b$}
\psfrag{v'}{$v'$}
\psfrag{v}{$v$}
\psfrag{r}{$r$}
\psfrag{l1}{}
\psfrag{l2}{$\lfloor(1+\epsilon)^i\rfloor$}
\psfrag{px}{$p_{a,b}$}
\psfrag{a}{(i)}
\psfrag{b}{(ii)}
\centerline{\epsfysize=110pt \epsfbox{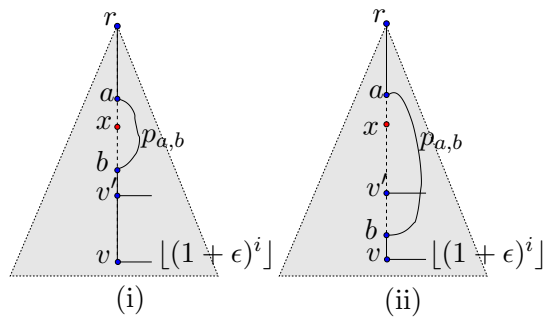}}
\caption{$p_{a,b}$ is shortest detour of $\textbf{P}(r,v,x)$.
(i) : detour of type I, (ii) : detour of type II}
\label{fig:detour-type}
\end{figure}

Handling detours of type I is relatively easy. Let $w$ be the farthest
ancestor of $v$ such that $w \in S$ and level of $w$ is greater or equal to
the level of $b$. In this case, $v$ stores the corresponding detour implicitly
by just keeping a pointer to the vertex $w$.

Handling detours of type II is slightly tricky since we can't afford to store
each of them explicitly. However, we shall employ the following observation
associated with the detours of type II to guarantee low space requirement.
\begin{obs}
Let ${\alpha}_1$, ${\alpha}_2$,\ldots,${\alpha}_t$ be the vertices on $\textbf{P}(r,v)$
(in the increasing order of their levels) such that the shortest detour
corresponding to $\textbf{P}(r,v,{\alpha}_i)$ is of type II $\forall i$, then
\[
\delta(r,v,{\alpha}_1) \ge \delta(r,v,{\alpha}_2)\ge \cdots\ge \delta(r,v,{\alpha}_t)
\]
\end{obs}
It follows from the above observation that if
$\delta(r,v,\alpha_i)\le (1+{\epsilon})\delta(r,v,\alpha_j)$ for any $i<j$,
then $\textbf{P}(r,v,{\alpha}_i)$ may as well serve as
$(1+{\epsilon})$-approximate shortest path from $r$ to $v$ avoiding $\alpha_j$.
In other words, we need not store the detour associated with
$\textbf{P}(r,v,{\alpha}_j)$ in such situation.
Using this observation, we shall have to explicitly store only
$O(\log_{1+{\epsilon}} n)$ detours of type II.
Moreover, we do not store explicitly detours of type II whose length is much
larger than $level(v)$. Specifically, if
$\textbf{P}(r,v,x) \geq \frac{1}{{\epsilon}}level(v)$, then 
$v$ will merely store pointer to the path
$\textbf{P}(r,{\tt uchild}(x),x)::\textbf{P}({\tt uchild}(x),v)$.
This ensures that each detour of type II which $v$ has to
explicitly store will have length $O(\frac{1}{{\epsilon}}level(v))$.

It follows from the above description that for a special vertex
$v$ and $x\in \textbf{P}(r,v)$, the data structure associated with $v$ stores
$(1+2\epsilon)$-approximation of the path $\textbf{P}(r,v,x)$. Moreover, the total
space required by the data structure associated with all the special vertices
will be $O(n\frac{\log n}{\epsilon^3})$. This supplementary data
structure combined with the data structure of previous section can report
$(1+6\epsilon)$-approximation of $\textbf{P}(r,z,x)$ for any $z,x\in V$.

\begin{theorem}
Given an undirected unweighted graph $G=(V,E)$, source $r\in V$, and any 
$\epsilon>0$, we can build a data structure of size
$O(n\frac{\log n}{\epsilon^3})$ 
that can report $(1+\epsilon)$-approximate shortest path from $r$ to any 
$z\in V$ avoiding any failed vertex in optimal time.
\label{main-theorem-1+epsilon}
\end{theorem}

\section{All-pairs $(2k-1)(1+\epsilon)$-approx. distance oracle avoiding a 
failed vertex}
We start with a brief description of the approximate distance oracle of
Thorup and Zwick \cite{TZ:5}. The key idea to achieve sub-quadratic space is
to store distance from each vertex to only a small set of vertices. For
retrieving approximate distance between any two vertices $u,v\in V$, it is
ensured that there is a third vertex $w$ which is {\em close} to both of them,
and whose distance from both of them is known.
To realize this idea, Thorup and Zwick \cite{TZ:5} introduced two novel
structures called {\em ball} and {\em cluster} which are defined for any two
subsets $A,B$ of vertices as follows. 
(here $\delta(v,B)$ denotes the distance between $v$ and its nearest vertex from
$B$).
\[
Ball(v,A,B) = \{w\in A | \delta(v,w) < \delta(v,B)\} \hspace*{1cm}
C(w,A,B) = \{v\in V | \delta(v,w) < \delta(v,B) \}
\]
Construction of $(2k-1)$-approximate distance oracle of Thorup and Zwick
\cite{TZ:5} employs a $k$-level hierarchy ${\mathbf A}_k=\langle A_0\supseteq A_1 \supseteq A_2 ... \supseteq A_{k-1} \supset A_k\rangle$ of subsets of vertices
as follows.\\
$A_0=V$, $A_k=\emptyset$, and $A_{i+1}$ for any $i <k-1$ is formed by
selecting each vertex from $A_i$ independently with probability $n^{-1/k}$.

The data structure associated with the $(2k-1)$-approximate distance oracle of
Thorup and Zwick \cite{TZ:5} stores for each vertex $v\in V$ the following
information :
\begin{itemize}
\item the vertices of set $\cup_{i<k} Ball(v,A_i,A_{i+1})$ (and their distances).
\item the vertex from $A_i$ nearest to $v$ (to be denoted as $p_i(v)$).
\end{itemize}
Due to randomization underlying the construction of ${\mathbf A_k}$, the
expected size of $Ball(v,A_i,A_{i+1})$ is $O(n^{1/k})$, and hence the space
required by the oracle is $O(kn^{1+1/k})$.
We shall now outline the ideas in extending the $(2k-1)$-approximate distance
oracle to handle single vertex failure. Kindly refer to the extended version 
\cite{KB:09} of this paper for complete details.

\subsection{Overview of all-pairs approx. distance oracles avoiding a failed
vertex}
Firstly the notations used by the static approximate distance oracle of
\cite{TZ:5}, in particular ball and cluster, get extended
for single vertex failure in a natural manner as follows. 
(here $\delta(v,B,x)$ is the distance between $v$ and its nearest vertex 
from $B$ in $G\backslash\{x\}$).
\[ Ball^x(v,A,B)
  = \{w\in A | \delta(v,w,x) < \delta(v,B,x)\}
\]
\[ C^x(w,A,B)
  = \{v\in V | \delta(v,w,x) < \delta(v,B,x)\}
\]
Let $p_i^x(v)$ denote the vertex from $A_i$ which is nearest to $v$ in
$G\backslash \{x\}$.
Along the lines of the static approximate distance oracle of Thorup and Zwick
\cite{TZ:5}, the basic operation which the approximate distance oracle 
avoiding a failed vertex should support is the following :

{\em Report distance (exact or approximate) between $v$ and $w\in A_i$ if
$w \in Ball^x(v,A_i,A_{i+1})$ for any given $v,x\in V$.}

However, it can be observed that we would have to support this operation
implicitly instead of explicitly keeping $Ball^x(v,A_i,A_{i+1})$ for each
$v,x,i$. Our starting point is the simple observation that clusters and balls
are inverses of each others, that is,
$w\in Ball^x(v,A_i,A_{i+1})$ is equivalent to $v\in C^x(w,A_i,A_{i+1})$.
Now we make an important observation. Consider the subgraph ${\mathbf G_i}(w)$
induced by the vertices of set $\cup_{x\in V} C^x(w,A_i,A_{i+1})$. This subgraph
preserves the path $\textbf{P}(w,v,x)$ for each $x,v\in V$
if $w\in Ball^x(v,A_i,A_{i+1})$. So it suffices to keep a single source
(approximate) shortest paths oracle on ${\mathbf G_i}(w)$ with $w$ as the root.
Keeping this data structure for each $w\in A_i$ provides an implicit
compact data structure for supporting the basic operation mentioned above. Using
Theorem \ref{main-theorem-1+epsilon}, it can be seen that the space required
at a level $i$ will be of the order of
$\sum_{w\in A_i} |\cup_{x\in V} C^x(w,A_i,A_{i+1})|$, but  it is not clear whether we can get an upper bound of the
order of $n^{1+1/k}$ on this quantity. Here, as a new tool, we introduce the notion of $\epsilon$-truncated balls and
clusters.
\begin{definition}
Given a vertex $x$, any subsets $A,B$, and $\epsilon>0$
\[ Ball^x(v,A,B, \epsilon)
  = \left\{w\in A | \delta(v,w,x) < \frac{\delta(v,B,x)}{1+\epsilon} \right\}
\]

\label{epsilon-truncated}
\end{definition}
Instead of dealing with the usual balls (and clusters) under deletion of single
vertex, we deal with $\epsilon$-truncated balls (and clusters) under deletion
of single vertex. We note that  the inverse relationship between clusters and balls gets seamlessly extended  to $\epsilon$-truncated balls and clusters under single vertex failure as well. That is,
\[
   \sum_{w\in A_i} |\cup_{x\in V} C^x(w,A_i,A_{i+1},\epsilon)| = \sum_{v\in V} |\cup_{x\in V} Ball^x(v,A_i,A_{i+1},\epsilon)|
\]
So it suffices to get an upper bound on the size of the set
$\cup_{x\in V} Ball^x(v,A_i,A_{i+1},\epsilon)$ for any vertex $v\in V$. The following lemma states a very crucial property of
$\epsilon$-truncated balls which leads to prove the existence of a small set
$S$ of $O(\frac{1}{\epsilon^2}\log n)$ vertices such that
\begin{equation}
 \cup_{x\in V} Ball^x(v,A_i,A_{i+1},\epsilon) \subseteq
\cup_{x\in S} Ball^x(v,A_i,A_{i+1})
\cup Ball(v,A_i,A_{i+1})
\label{for-S}
\end{equation}
%
\begin{lemma}
In a given graph $G=(V,E)$, let $v$ be any vertex and let $u=p_{i+1}(v)$.
Let $x_1$ and $x_2$ be any two vertices on the $\textbf{P}(v,u)$ path with
$x_1$ appearing closer to $v$ on this path and
$\delta(v,A_{i+1},x_1) \le (1+\epsilon) \delta(v,A_{i+1},x_{2})$. Then\\
\hspace*{1in}
$
Ball^{x_1}(v,A_i,A_{i+1},\epsilon) \subseteq Ball(v,A_{i},A_{i+1})\cup Ball^{x_{2}}(v,A_i,A_{i+1})
$
\label{x1_and_x2}
\end{lemma}
\begin{proof}
Let $w$ be any vertex in $A_i$. It suffices to show the following. If $w$
does not belong to $Ball(v,A_{i},A_{i+1})\cup Ball^{x_{2}}(v,A_i,A_{i+1})$, then
$w$ does not belong to $Ball^{x_1}(v,A_i,A_{i+1},\epsilon)$.
The proof is based on the analysis of the following two cases.\\
{\bf Case 1 : ~The vertex $x_2$ is present in $\textbf{P}(v,w,x_1)$.}\\
Since, $w\notin Ball(v,A_i,A_{i+1})$, therefore, $\delta(v,w)$ is at least
$\delta(v,u)$. Hence using triangle inequality, $\delta(v,x_2) + \delta(x_2,w) \ge \delta(v,u)$. Now $\delta(v,u)=\delta(v,x_2) + \delta(x_2,u)$ (since $x_2$
lies on $P(v,u)$). Hence $\delta(x_2,w) \ge \delta(x_2,u)$. Moreover, since
$x_1$ does not appear on $\textbf{P}(x_2,u)$, so
$\delta(x_2,u)= \delta(x_2,u,x_1)$. So
\begin{equation}
\delta(x_2,w,x_1)\ge \delta(x_2,u,x_1)
\label{case-1-apasp}
\end{equation}
Now it is given that $x_2\in \textbf{P}(v,w,x_1)$, so $\textbf{P}(v,w,x_1)$
must be of the form $\textbf{P}(v,x_2,x_1)::\textbf{P}(x_2,w,x_1)$, the length
of which is at least $\delta(v,x_2,x_1) + \delta(x_2,u,x_1)$ using Equation
\ref{case-1-apasp}. The latter quantity is at least $\delta(v,u,x_1)$ which
by definition is at least $\delta(v,A_{i+1},x_1)$.
Hence $w\notin Ball^{x_1}(v,A_i,A_{i+1})$, and therefore,
$w\notin Ball^{x_1}(v,A_i,A_{i+1},\epsilon)$.\\
{\bf Case 2 : ~The vertex $x_2$ is not present in $\textbf{P}(v,w,x_1)$.}\\
In this case, $\delta(v,w,x_1) =  \delta(v,w,\{x_1,x_2\}) \ge \delta(v,w,x_2)$.
The value $\delta(v,w,x_2)$ is in turn at least $\delta(v,A_{i+1},x_2)$
since $w\notin Ball^{x_{2}}(v,A_i,A_{i+1})$. It is given that  $\delta(v,A_{i+1},x_2) \ge \frac{\delta(v,A_{i+1},x_1)}{1+\epsilon}$, hence conclude that
$\delta(v,w,x_1)\ge \frac{\delta(v,A_{i+1},x_1)}{1+\epsilon}$.
So $w\notin Ball^{x_1}(v,A_i,A_{i+1},\epsilon)$.\\
\end{proof}

We shall now outline the construction of a small set
$S$ of vertices which will satisfy Equation \ref{for-S}.
Let $u=p_{i+1}(v)$ and let $\textbf{P}(v,u)=v(=x_0),x_1,...,x_{\ell}(=u)$.
Observe that $\cup_{x\in V} Ball^x(v,A_i,A_{i+1},\epsilon) =
\cup_{1\le j\le \ell} Ball^{x_j}(v,A_i,A_{i+1},\epsilon)$.
For any node $x\in {\mathbf P}(u,v)$, let $value(x)=\delta(v,A_{i+1},x)$, and
let $h$ be the maximum $value$ of any node on this path. The set $S$ is
initially empty.

Let $\alpha(1)$ be the largest index from $[1,\ell]$ such that
$value(x_i)\ge h/(1+\epsilon)$. It can be seen that for all $j<\alpha(1)$,
$\delta(v,A_{i+1},x_j)\le(1+\epsilon)\delta(v,A_{i+1},x_{\alpha(1)})$.
Therefore, it follows from Lemma \ref{x1_and_x2} that for each vertex
$x\in \{x_1,...,x_{\alpha(1)}\}$,
$Ball^{x}(v,A_i,A_{i+1},\epsilon) \subseteq
Ball^{x_{\alpha(1)}}(v,A_i,A_{i+1})\cup Ball(v,A_i, A_{i+1})$.
So we insert $x_{\alpha(1)}$ to $S$. Similarly $\alpha(2)\in [\alpha(1)+1,\ell]$
be the greatest integer such that $value(x_{\alpha(2)})\ge h/(1+\epsilon)^2$. We
add $x_{\alpha(2)}$ to $S$, and so on. It can be seen that the set
$S$ constructed in this manner will satisfy Equation \ref{for-S} and its
size will be $O(\log_{1+\epsilon} h)=O(\frac{\log n}{\epsilon})$.

It can be shown using elementary probability theory that for each $x\in V$,
the set $Ball^x(v,A_i,A_{i+1})$ has size $O(n^{1/k}\log n)$ with high
probability. Therefore, the construction of the set $S$ outlined above
implies the following crucial bound for each $v\in V, i<k-1$ which helps us
design all-pairs approximate distance oracle avoiding a failed vertex.
\[ |\cup_{x\in V} Ball^x(v,A_i,A_{i+1},\epsilon)| =  O\left(n^{1/k} \frac{\log^2 n}{\epsilon}\right)\]
Using this equation, and owing to inverse relationship between clusters and
balls, it follows that $\sum_{w\in A_i}|\cup_{x \in V} C^x(w,A_i,A_{i+1},\epsilon)| = O\left(n^{1+1/k} \frac{\log^2 n}{\epsilon}\right)$.
Our all-pairs approximate distance oracle avoiding any failed vertex will keep
the following data structures.
\begin{itemize}
\item
Let $p_i^x(v,\epsilon)$ denote a vertex $w$ from $A_i$ with
$\delta(v,w,x)\le (1+\epsilon) \delta(v,p_i^x(v),x)$.
We keep a data structure $\textbf{N}_i$  $\forall i <k$, using which we can
retrieve $p_i^x(v,\epsilon)$. This data-structure is obtained by suitable
augmentation of our single source $(1+\epsilon)$-approximate oracle.
\item
For each $w\in A_i$, we  keep our single source $(1+\epsilon)$-approximate
oracle in $\textbf{G}_i(w,\epsilon)$ which
is the subgraph induced by $\cup_{x\in V} C^x(w,A_i,A_{i+1},\epsilon)$.
\end{itemize}
It follows that the overall space required by the data structure will be
$O(kn^{1+1/k} \frac{\log^3n}{\epsilon^4})$. The query
algorithm and the analysis on the stretch of the approximate distance
reported by the oracle are similar in spirit to that of
Thorup and Zwick \cite{TZ:5} (see \cite{KB:09} for details).
\begin{theorem}
Given an integer $k>1$ and a fraction $\epsilon>0$, an unweighted graph
$G=(V,E)$ can be processed to construct a data structure which can answer
$(2k-1)(1+\epsilon)$-approximate distance query between any two nodes
$u\in V$ and $v\in V$ avoiding any single failed vertex in $O(k)$ time.
The total size of the data structure is
$O(kn^{1+1/k}\frac{\log^3 n}{\epsilon^4})$.
\label{APASP-oracle-under-vertex-failure}
\end{theorem}

\noindent
{\large \bf Future work.}~
{\bf (i)} 
Can we design a data structure for single source $(1+\epsilon)$-approx.
shortest paths avoiding a failed vertex for weighted graphs ? Such a data 
structure will immediately extend our all-pairs approx. distance oracle 
avoiding a failed vertex to weighted graphs.\\ 
{\bf (ii)} How to design approx. distance oracles avoiding two or
more failed vertices ? Recent work of Duan and Pettie \cite{DP:9}, and Chechik
et al. \cite{CLDR:9} provides additional motivation for this.

\bibliographystyle{abbrv}
\bibliography{khanna}

\end{document}